\documentclass[journal,a4]{IEEEtran}
\usepackage[T1]{fontenc}

\usepackage{subfigure}
\usepackage{amsmath}
\usepackage{graphicx}
\usepackage{epstopdf}
\DeclareGraphicsExtensions{.eps,.pdf} 
\graphicspath{ {figures/} }
\usepackage{cite}
\usepackage{subfigure}
\usepackage{algorithm}
\usepackage{algorithmic}
\usepackage{xcolor}
\usepackage{multirow}
\usepackage{multicol}
\usepackage{cases}
\usepackage{setspace}
\usepackage[cmintegrals]{newtxmath}
    \usepackage[compact]{titlesec}
    \titlespacing{\section}{0pt}{2ex}{1ex}
    \titlespacing{\subsection}{0pt}{1ex}{0ex}
    \titlespacing{\subsubsection}{0pt}{0.5ex}{0ex}

\newcommand{\diag}{{\mathrm{diag}}}
\newcommand{\st}{{\mathrm{s.t.}}}
\newcommand{\ds}{\displaystyle}
\newtheorem{theorem}{\bf {Theorem}}

\begin{document}
%
\title{\huge Energy Efficiency Maximization in RIS-Aided Cell-Free Network with Limited Backhaul  }
%
%
%

\author{Quang~Nhat~Le,~\IEEEmembership{Graduate~Student~Member,~IEEE,}
        Van-Dinh~Nguyen,~\IEEEmembership{Member,~IEEE,}
        Octavia~A.~Dobre,~\IEEEmembership{Fellow,~IEEE,}
        and~Ruiqin~Zhao,~\IEEEmembership{Member,~IEEE}
\vspace{-20pt}
\thanks{This work was supported by the Natural Sciences and Engineering Research Council of Canada (NSERC) through its Discovery program.}
\thanks{Q. N. Le and O. A. Dobre are with Department of Electrical and Computer Engineering, Memorial University, St. John’s, NL A1B 3X9, Canada (e-mail: \{qnle, odobre\}@mun.ca).}
\thanks{V.-D. Nguyen is with the Interdisciplinary Centre for Security, Reliability, and Trust (SnT) – University of Luxembourg, L-1855, Luxembourg (e-mail: dinh.nguyen@uni.lu).}
\thanks{R. Zhao is with Key Laboratory of Ocean Acoustics and Sensing, School of Marine Science and Technology, Northwestern Polytechnical University, Xi’an 710072, China (e-mail: rqzhao@nwpu.edu.cn).}
\vspace{-10pt}
}

%
%

\markboth{}%
{Shell \MakeLowercase{\textit{et al.}}: Bare Demo of IEEEtran.cls for IEEE Communications Society Journals}
%



\maketitle

\begin{abstract}
Integrating the reconfigurable intelligent surface in a cell-free (RIS-CF) network is an effective solution to improve the capacity and coverage of future wireless systems with  low cost and power consumption. The reflecting coefficients of RISs can be programmed to enhance  signals received at users. This letter addresses a joint design of transmit beamformers at access points and reflecting coefficients at RISs to maximize the energy efficiency (EE) of RIS-CF networks, taking into account the limited backhaul capacity constraints. Due to a very computationally challenging nonconvex problem, we develop a simple yet efficient alternating descent algorithm for its solution. Numerical results verify that the EE of RIS-CF networks is greatly improved, showing the benefit of using RISs.
\end{abstract}


\IEEEpeerreviewmaketitle

\section{Introduction}

\IEEEPARstart{U}{ltra}-dense networks (UDNs) have been advocated as a key enabler for beyond fifth-generation wireless networks to further increase  network capacity \cite{YadavWC2018}. The underlying principle of UDN is to densely deploy a large number of access points (APs) and small cells in  cellular networks. However, the high density of APs and small cells comes at a cost of severe inter-cell interference \cite{TengST2018}. 

In order to address this bottleneck, cell-free (CF) networks have been recently proposed as a promising technology to effectively resolve the interference issues in  existing cellular networks \cite{HieuJSAC2020,Ngo:TWC:Mar2017}. Since each user equipment (UE) in the network is coherently served by a large number of APs coordinated by a central processing unit (CPU) with no cell boundaries, inter-cell interference can be efficiently reduced, and thus the network capacity can be enhanced accordingly \cite{Ngo:TWC:Mar2017}. Nonetheless, the performance of CF networks is heavily constrained by the limited backhaul capacity between  APs and  CPU \cite{LuongTSP17,PolegreWSA20}. Further, the dense deployment of APs in CF networks results in an increase in the network energy consumption \cite{Ngo:IEEETGCN:Mar2018}. Therefore, an efficient scheme to improve the network energy efficiency (EE), which will be considered as a major figure-of-merit in the design of future networks, is of crucial importance. 

Fortunately, the new revolutionary technology called reconfigurable intelligent surface (RIS) has been identified as a spectral efficient solution with low cost and power consumption \cite{WuTWC20}. An RIS consists of a large number of low-cost passive elements, where each element can be adjusted with an independent phase shift to reflect the electromagnetic incident signals, to be added coherently at UEs. It is not too far-fetched to envision a wireless system integrating RIS in a CF network, referred to as RIS-CF, reaping all key advantages of these two technologies. Despite its potential, only some attempts have been made to characterize the performance of RIS-CF in the literature \cite{Zhang2020,Huang2020}.  Unlike these works, which are mainly focusing on maximizing the sum-rate  with infinite backhaul capacity links, our goal is to achieve an optimal tradeoff between the total sum-rate and power consumption, taking into account the impact of limited backhaul capacity. 

Naturally, the beamformers at APs and RIS reflecting coefficients need to be jointly optimized to maximize the EE of RIS-CF, which results in a computationally intractable problem since the optimization variables are strongly coupled. To efficiently solve this problem, the alternating descent-based iterative algorithm is proposed, which converges at least to a locally optimal solution. In each iteration  of alternating optimization, we develop new approximate functions to tackle the nonconvex parts by leveraging the inner approximation (IA) framework \cite{Beck:JGO:10} and introducing a novel penalty function. Simulation results confirm that the proposed algorithm greatly improves the EE of cell-free networks over the existing approaches.

\emph{Notation}: $\mathbf{X}^{T}$ and $\mathbf{X}^{H}$   are the transpose and Hermitian transpose  of a matrix $\mathbf{X}$, respectively.  $\|\cdot\|$ and $|\cdot|$ denote the Euclidean norm of a vector and the absolute value of a complex scalar, respectively.  $\Re\{\cdot\}$ returns the real part of an argument.   

\section{System Model} \label{sec:sys}
\begin{figure}[!h]
	\centering	\vspace{-10pt}
	\includegraphics[width=0.30\textwidth,trim={0cm 0.0cm 0cm 0.0cm}]{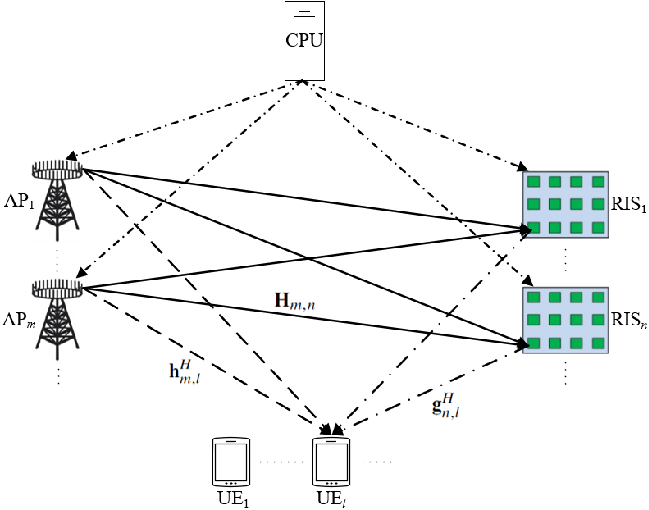}
	\vspace{-5pt}
	\caption{Illustration of an RIS-CF network.}
	\label{fig:sys1}
\end{figure}

We consider an RIS-CF network as illustrated in Fig. \ref{fig:sys1}, where the sets $\mathcal{M}\triangleq\{1,2,\cdots,M\}$ of $M$ APs and $\mathcal{N}\triangleq\{1,2,\cdots,N\}$ of $N$ RISs are distributedly deployed to coherently serve the set $\mathcal{L}\triangleq\{1,2,\cdots,L\}$ of $L$ single-antenna UEs. Each AP is equipped with $K$ antennas, and each RIS is composed of the set $\mathcal{R}\triangleq\{1,2,\cdots,R\}$ of $R$ passive reflecting elements. A CPU is deployed for control and planning purposes, to which all APs are connected by wired limited-capacity backhaul links. The backhaul link between AP$_m$ and CPU  has the predetermined maximum capacity $C_{m}^{\max}$, $\forall m \in \mathcal{M}$. All RISs are controlled by the CPU or APs by wired or wireless links.

\subsection{Transmission Model}
The transmitted complex baseband signal $\boldsymbol{x}_m \in \mathbb{C}^{K \times 1}$ at AP$_m$ can be written as
$\boldsymbol{x}_m = \sum_{l\in\mathcal{L}} \boldsymbol{w}_{m,l} s_{l},$
where $s_{l}$ with $\mathbb{E}\{|s_{l}|^2\} =1$ and $\boldsymbol{w}_{m,l} \in \mathbb{C}^{K \times 1}$ are the transmitted symbol and beamforming vector intended for  UE $l$, respectively.
Due to the directional reflection supported by $N$ RISs, the channel between an AP and a UE  includes two parts: the AP-UE (direct) link and $N$ AP-RIS-UE (reflected) links. The equivalent channel $\hat{\mathbf{h}}^{H}_{m,l} \in \mathbb{C}^{1 \times K}$ from 
AP$_m$ to  UE $l$ can be expressed as
\begin{align}
\hat{\textbf{h}}^{H}_{m,l}(\boldsymbol{\psi}) & = \textbf{h}^{H}_{m,l} + \sum\nolimits_{n\in\mathcal{N}} \textbf{g}^{H}_{n,l} \boldsymbol{\Phi}_{n} \textbf{H}_{m,n} \nonumber \\ 
&= \textbf{h}^{H}_{m,l} + \sum\nolimits_{n\in\mathcal{N}} \boldsymbol{\psi}^{T}_{n} \diag \bigl(\textbf{g}^{H}_{n,l}\bigr) \textbf{H}_{m,n}
\end{align}
where $\textbf{h}^{H}_{m,l} \in \mathbb{C}^{1 \times K}$, $\textbf{H}_{m,n} \in \mathbb{C}^{R \times K}$, and $\textbf{g}^{H}_{n,l} \in \mathbb{C}^{1 \times R}$ denote the channels from  AP$_m$ to  UE $l$, from  AP$_m$ to  RIS$_{n}$, and from  RIS$_{n}$ to  UE $l$, respectively. $\boldsymbol{\Phi}_{n} \in \mathbb{C}^{R \times R}$ represents the phase shift matrix of  RIS$_{n}$, which can be written as \cite{WuTWC20}:
$\boldsymbol{\Phi}_{n} \triangleq \diag(e^{j\theta_{n,1}},e^{j\theta_{n,2}},\ldots,e^{j\theta_{n,R}}),$
where $\theta_{n,r} \in [0,2\pi)$ denotes the phase shift of the $r$-th reflecting element on the RIS$_{n}$. Further, $\boldsymbol{\Phi}_{n}$ can be rewritten as $\boldsymbol{\Phi}_{n} = \diag\left(\psi_{n,1},\psi_{n,2},\ldots,\psi_{n,R}\right)$, with $|\psi_{n,r}| = 1$, $\forall n \in \mathcal{N}$, $ r \in \mathcal{R}$. Let us define $\boldsymbol{\psi} \triangleq \{\boldsymbol{\psi}_{n}\}_{\forall n}$ with $\boldsymbol{\psi}_{n} = \bigr[ \psi_{n,1},\psi_{n,2},\ldots,\psi_{n,R} \bigr]^{T}$.

The signal  received at  UE $l$ can be expressed as
\begin{align}\label{eq:yli}
y_{l} & = \sum\nolimits_{m\in\mathcal{M}} \hat{\textbf{h}}^{H}_{m,l} \boldsymbol{x}_m + n_{l}
\end{align}
where $n_{l} \sim \mathcal{CN}(0,\sigma^2)$ is the additive white Gaussian noise (AWGN) noise at  UE $l$. The achievable data rate (nats/s/Hz) of UE $l$ is given as
\begin{align}\label{eq:rate}
{R}_{l} (\boldsymbol{w},\boldsymbol{\psi}) &= \ln \left(1 + \gamma_{l} (\boldsymbol{w},\boldsymbol{\psi}) \right) \nonumber \\
&= \ln \Biggl( 1 + \frac{\bigl| \sum\nolimits_{m\in\mathcal{M}} \hat{\textbf{h}}^{H}_{m,l} (\boldsymbol{\psi}) \boldsymbol{w}_{m,l} \bigl|^2}{\sum\nolimits_{j \in \mathcal{L} \setminus l} \bigr| \sum\nolimits_{m\in\mathcal{M}} \hat{\textbf{h}}^{H}_{m,l} (\boldsymbol{\psi}) \boldsymbol{w}_{m,j} \bigr|^2 + \sigma^2} \Biggr) \nonumber \\
&= \ln \Biggl( 1 + \frac{\bigl| \hat{\textbf{h}}^{H}_{l} (\boldsymbol{\psi}) \boldsymbol{w}_{l} \bigl|^2}{\sum\nolimits_{j \in \mathcal{L} \setminus l} \bigl| \hat{\textbf{h}}^{H}_{l} (\boldsymbol{\psi}) \boldsymbol{w}_{j} \bigr|^2 + \sigma^2} \Biggr)
\end{align}
where $\boldsymbol{w} \triangleq \{\boldsymbol{w}_{m,l}\}_{\forall m, l}$,  $\hat{\textbf{h}}_{l} = \bigl[ \hat{\textbf{h}}^{H}_{1,l}, \hat{\textbf{h}}^{H}_{2,l},\ldots, \hat{\textbf{h}}^{H}_{M,l} \bigr]^{H}$, and $\boldsymbol{w}_{l} = \bigr[ \boldsymbol{w}^{H}_{1,l}, \boldsymbol{w}^{H}_{2,l},\ldots, \boldsymbol{w}^{H}_{M,l} \bigr]^{H}$. 


\subsection{Optimization Problem Formulation}

\textbf{Power consumption model:} The total power consumption of the proposed RIS-CF network is modeled as
\begin{IEEEeqnarray}{rCl}
{P}_{\Sigma}(\boldsymbol{w}) &=&\sum\nolimits_{m\in\mathcal{M}} \xi_{m} \sum\nolimits_{l\in\mathcal{L}} \|\boldsymbol{w}_{m,l}\|^2 +\sum\nolimits_{m\in\mathcal{M}} {P}_{m}   \nonumber \\
&+& \sum\nolimits_{l\in\mathcal{L}} {P}_{l}+ \sum\nolimits_{n\in\mathcal{N},r\in\mathcal{R}}  {P}_{n,r} + \sum\nolimits_{m\in\mathcal{M},l\in\mathcal{L}}{P}^{\rm{BH}}_{m,l}\qquad
\end{IEEEeqnarray}
where ${P}_{m}$ and ${P}_{l}$ denote the circuit power consumption of  AP$_{m}$ and  UE $l$, respectively. $\xi_{m}$ regulates the ineffectiveness of the power amplifier at AP$_{m}$, and ${P}_{n,r}$ represents the low-power consumption of the $r$-th  reflecting element in the $n$-th RIS \cite{HuangTWC19}. The power consumption for conveying the data and beamformers related to the transmission from  AP$_m$ to  UE $l$ via backhaul transmission is represented by ${P}^{\rm{BH}}_{m,l}$.

\textbf{Backhaul constraint:} The data rate transmitted by the $m$-th backhaul link should be $\omega_{m}$ times greater than or equal to the total achievable rate at  AP$_{m}$, with $\omega_{m} \geq 1$, $\forall m \in \mathcal{M}$ \cite{LuongTSP17,PolegreWSA20}. Then, the per-backhaul capacity constraints can be expressed as:
\begin{align}
\sum\nolimits_{l\in\mathcal{L}} {R}_{l} (\boldsymbol{w},\boldsymbol{\psi}) \leq \frac{C_{m}^{\max}}{\omega_{m}},\ \forall m \in \mathcal{M}.
\end{align} 

Our goal is to maximize the EE of the RIS-CF network by jointly optimizing the involved variables $(\boldsymbol{w},\boldsymbol{\psi})$, stated as
\begin{subequations}\label{eq:op4}
	\begin{align}
	 \underset{\boldsymbol{w},\boldsymbol{\psi}}{\max}& \label{eq:op4a}
	 \quad\mathcal{E}(\boldsymbol{w},\boldsymbol{\psi})\triangleq\frac{B \sum_{l\in\mathcal{L}} {R}_{l}(\boldsymbol{w},\boldsymbol{\psi})}{{P}_{\Sigma}(\boldsymbol{w})} \\ \label{eq:op4b}
	\st&  \quad  \sum\nolimits_{l\in\mathcal{L}} \|\boldsymbol{w}_{m,l}\|^2 \leq {P}^{\max}_{m}, \forall m \in \mathcal{M}, \\ \label{eq:op4c}
	&\quad  {R}_{l}(\boldsymbol{w},\boldsymbol{\psi}) \geq {R}^{\rm{min}}_{l}, \forall l \in \mathcal{L}, \\ \label{eq:op4d}
	&\quad \sum\nolimits_{l\in\mathcal{L}} {R}_{l}(\boldsymbol{w},\boldsymbol{\psi}) \leq \frac{C_{m}^{\max}}{\omega_{m}}, \forall m \in \mathcal{M}, \\ \label{eq:op4e}
	&\quad  |\psi_{n,r}| = 1, \forall n \in \mathcal{N}, r \in \mathcal{R}
	\end{align}
\end{subequations}
where \eqref{eq:op4b} indicates the power constraint at  AP$_{m}$ with the maximum transmit power ${P}^{\rm{max}}_{m}$ and constraint \eqref{eq:op4c} is imposed to guarantee the minimum achievable rate requirement ${R}^{\rm{min}}_{l}$ of  UE $l$. Problem \eqref{eq:op4}  is nonconvex since the objective is nonconcave and constraints \eqref{eq:op4c}-\eqref{eq:op4e} are nonconvex. The complex rate function in \eqref{eq:rate} and the nonconvex constraint on the reflecting coefficients \eqref{eq:op4e} make this problem even more challenging to solve jointly.

\section{Proposed Alternating Descent-based  Iterative Algorithm}
In an iterative algorithm based on the IA framework \cite{Beck:JGO:10},
let ($\boldsymbol{\psi}^{(\kappa)},\boldsymbol{w}^{(\kappa)}$) be the feasible point of \eqref{eq:op4} obtained at the ($\kappa - 1)$-th iteration. In this section, an alternating descent algorithm with low complexity is proposed to solve \eqref{eq:op4}, i.e. at iteration $\kappa+1$ solving \eqref{eq:op4} to find the optimal solution $\boldsymbol{w}^\star:=\boldsymbol{w}^{(\kappa+1)}$ for given $\boldsymbol{\psi}^{(\kappa)}$, and then solving \eqref{eq:op4} to find the optimal solution $\boldsymbol{\psi}^\star:=\boldsymbol{\psi}^{(\kappa+1)}$ for given $\boldsymbol{w}^{(\kappa+1)}$.

\subsection{Beamforming Descent Iteration}
At  iteration $\kappa+1$, problem \eqref{eq:op4} for given $\boldsymbol{\psi}^{(\kappa)}$ can be expressed as
\begin{subequations} \label{eq:op4w}
	\begin{IEEEeqnarray}{cl}
		\underset{\boldsymbol{w},\rho}\max\quad  &
		\mathcal{F}(\boldsymbol{w},\boldsymbol{\psi}^{(\kappa)})\triangleq\frac{B \sum_{l\in\mathcal{L}} {R}_{l}(\boldsymbol{w}|\boldsymbol{\psi}^{(\kappa)})}{\rho} \label{eq:op4wa}\\ \st\quad & {P}_{\Sigma}(\boldsymbol{w}) \leq \rho, \label{eq:op4wa1}\\ 
		&   {R}_{l}(\boldsymbol{w}|\boldsymbol{\psi}^{(\kappa)}) \geq {R}^{\rm{min}}_{l}, \forall l \in \mathcal{L}, \label{eq:op4wc}\\ 
		&   \sum\nolimits_{l\in\mathcal{L}} {R}_{l}(\boldsymbol{w}|\boldsymbol{\psi}^{(\kappa)}) \leq \frac{C_{m}^{\max}}{\omega_{m}}, \forall m \in \mathcal{M}, \label{eq:op4wd}\\
		& \eqref{eq:op4b}
	\end{IEEEeqnarray}
\end{subequations}
where $\rho$ is a slack variable to represent the soft power consumption of RIS-CF. The objective \eqref{eq:op4wa} is nonconcave, and constraints \eqref{eq:op4wc} and \eqref{eq:op4wd} are nonconvex in $\boldsymbol{w}$. To tackle the nonconcavity of \eqref{eq:op4wa}, we use the following inequality: 
\begin{align}\label{eq:ieObj}
	\frac{1}{z}\ln\bigl(1+x^2/y\bigr) \geq \mathcal{A}^{(\kappa)} - \mathcal{B}^{(\kappa)}\frac{y}{x^2} - \mathcal{C}^{(\kappa)}z, \ \forall x,y, z\in\mathbb{R}_+
\end{align}
where $\mathcal{A}^{(\kappa)} \triangleq 2\frac{\ln\bigl(1+(x^{(\kappa)})^2/y^{(\kappa)}\bigr)}{z^{(\kappa)}} + \frac{(x^{(\kappa)})^2/y^{(\kappa)}}{z^{(\kappa)}(1+(x^{(\kappa)})^2/y^{(\kappa)})}$, $\mathcal{B}^{(\kappa)} \triangleq \frac{\bigl((x^{(\kappa)})^2/y^{(\kappa)}\bigr)^2}{z^{(\kappa)}(1+(x^{(\kappa)})^2/y^{(\kappa)})}$, and $\mathcal{C}^{(\kappa)} \triangleq \frac{\ln\bigl(1+(x^{(\kappa)})^2/y^{(\kappa)}\bigr)}{(z^{(\kappa)})^2}$. The proof of \eqref{eq:ieObj} is given in Appendix A.
For $\bar{\boldsymbol{w}}_{l} =  e^{-j\arg(\hat{\textbf{h}}^{H}_{l} (\boldsymbol{\psi}^{(\kappa)}) \boldsymbol{w}_{l})}\boldsymbol{w}_{l}$ with $j=\sqrt{-1}$, it follows that $|\hat{\textbf{h}}^{H}_{l} (\boldsymbol{\psi}^{(\kappa)}) \boldsymbol{w}_{l} | = \hat{\textbf{h}}^{H}_{l} (\boldsymbol{\psi}^{(\kappa)}) \bar{\boldsymbol{w}}_{l} = \Re\{\hat{\textbf{h}}^{H}_{l} (\boldsymbol{\psi}^{(\kappa)}) \bar{\boldsymbol{w}}_{l}\} \geq 0$ and 
$|\hat{\textbf{h}}^{H}_{l} (\boldsymbol{\psi}^{(\kappa)}) \boldsymbol{w}_{l'} | = \hat{\textbf{h}}^{H}_{l} (\boldsymbol{\psi}^{(\kappa)}) \bar{\boldsymbol{w}}_{l'}$ for all $l'\neq l$. Thus, ${R}_{l} (\boldsymbol{w},\boldsymbol{\psi}^{(\kappa)})$ can be rewritten as
\begin{align}\label{eq:9rewrite}
{R}_{l} (\boldsymbol{w}|\boldsymbol{\psi}^{(\kappa)}) = \ln \Bigr( 1 + \frac{\bigl( \Re\{\hat{\textbf{h}}^{H}_{l} (\boldsymbol{\psi}^{(\kappa)}) \boldsymbol{w}_{l}\} \bigr)^2}{\varphi_l(\boldsymbol{w}|\boldsymbol{\psi}^{(\kappa)})} \Bigl)
\end{align}
under the condition that $\Re\{\hat{\textbf{h}}^{H}_{l} (\boldsymbol{\psi}^{(\kappa)}) \boldsymbol{w}_{l}\} \geq 0 $, where $\varphi_l(\boldsymbol{w}|\boldsymbol{\psi}^{(\kappa)}) \triangleq \sum\nolimits_{j \in \mathcal{L} \setminus l} | \hat{\textbf{h}}^{H}_{l} (\boldsymbol{\psi}^{(\kappa)}) \boldsymbol{w}_{j}|^2 + \sigma^2$. Applying inequality \eqref{eq:ieObj}, we obtain
\begin{IEEEeqnarray}{rCl}\label{eq:ieObjCo}
	\frac{ {R}_{l}(\boldsymbol{w}|\boldsymbol{\psi}^{(\kappa)})}{\rho} \geq
	\mathcal{A}_l^{(\kappa)} - \mathcal{B}_l^{(\kappa)}\frac{\varphi_l(\boldsymbol{w}|\boldsymbol{\psi}^{(\kappa)})}{\bigl( \Re\{\hat{\textbf{h}}^{H}_{l} (\boldsymbol{\psi}^{(\kappa)}) \boldsymbol{w}_{l}\} \bigr)^2}
	- \mathcal{C}_l^{(\kappa)}\rho\qquad
\end{IEEEeqnarray}
where $\mathcal{A}_l^{(\kappa)} \triangleq 2\ln\bigl(1+\Gamma_l^{(\kappa)}\bigr)/\rho^{(\kappa)} + \Gamma_l^{(\kappa)}/\bigl(\rho^{(\kappa)}(1+\Gamma_l^{(\kappa)})\bigr)$, $\mathcal{B}_l^{(\kappa)}\triangleq (\Gamma_l^{(\kappa)})^2/\bigl(\rho^{(\kappa)}(1+\Gamma_l^{(\kappa)})\bigr)$,  $\mathcal{C}_l^{(\kappa)}\triangleq \ln\bigl(1+\Gamma_l^{(\kappa)}\bigr)/(\rho^{(\kappa)})^2$, and $\Gamma_l^{(\kappa)} = \bigl( \Re\{\hat{\textbf{h}}^{H}_{l} (\boldsymbol{\psi}^{(\kappa)}) \boldsymbol{w}_{l}^{(\kappa)}\} \bigr)^2/\varphi_l(\boldsymbol{w}^{(\kappa)}|\boldsymbol{\psi}^{(\kappa)})$. As a result, the concave lower bound of $ {R}_{l}(\boldsymbol{w},\boldsymbol{\psi}^{(\kappa)})/\rho$ is found as
\begin{align}\label{eq:ieObjConave}
&\mathcal{F}_l^{(\kappa)}(\boldsymbol{w},\rho|\boldsymbol{\psi}^{(\kappa)}) :=  \mathcal{A}_l^{(\kappa)} - \mathcal{B}_l^{(\kappa)}\frac{\varphi_l(\boldsymbol{w}|\boldsymbol{\psi}^{(\kappa)})}{\Omega_l^{(\kappa)}(\boldsymbol{w}|\boldsymbol{\psi}^{(\kappa)})} 
 - \mathcal{C}_l^{(\kappa)}\rho
\end{align}
with the condition
 $\Omega_l^{(\kappa)}(\boldsymbol{w}|\boldsymbol{\psi}^{(\kappa)})\triangleq 2\Re\{\hat{\textbf{h}}^{H}_{l} (\boldsymbol{\psi}^{(\kappa)}) \boldsymbol{w}^{(\kappa)}_{l}\}$ $\Re\{\hat{\textbf{h}}^{H}_{l} (\boldsymbol{\psi}^{(\kappa)}) \boldsymbol{w}_{l}\} - \bigl(\Re\{\hat{\textbf{h}}^{H}_{l} (\boldsymbol{\psi}^{(\kappa)}) \boldsymbol{w}^{(\kappa)}_{l}\}\bigr)^2 > 0$. We note that $\mathcal{F}_l^{(\kappa)}(\boldsymbol{w},\rho|\boldsymbol{\psi}^{(\kappa)})$ is a concave lower bound of ${R}_{l}(\boldsymbol{w}|\boldsymbol{\psi}^{(\kappa)})/\rho$, satisfying  $\mathcal{F}_l^{(\kappa)}(\boldsymbol{w}^{(\kappa)},\rho^{(\kappa)},\boldsymbol{\psi}^{(\kappa)})={R}_{l}(\boldsymbol{w}^{(\kappa)},\boldsymbol{\psi}^{(\kappa)})/\rho^{(\kappa)}$.

Following the steps \eqref{eq:9rewrite}-\eqref{eq:ieObjConave} with $\rho=1$, constraint \eqref{eq:op4wc} can be directly convexified by 
\begin{equation}\label{eq:op4wcConvex}
\mathcal{R}_l^{(\kappa)}(\boldsymbol{w}|\boldsymbol{\psi}^{(\kappa)})
\geq {R}^{\rm{min}}_{l}, \forall l \in \mathcal{L}
\end{equation}
where $\mathcal{R}_l^{(\kappa)}(\boldsymbol{w}|\boldsymbol{\psi}^{(\kappa)})\triangleq \bar{\mathcal{A}}_l^{(\kappa)} - \bar{\mathcal{B}}_l^{(\kappa)} \frac{\varphi_l(\boldsymbol{w}|\boldsymbol{\psi}^{(\kappa)})}{\Omega_l^{(\kappa)}(\boldsymbol{w}|\boldsymbol{\psi}^{(\kappa)})}$, with $\bar{\mathcal{A}}_l^{(\kappa)} \triangleq \ln\bigl(1+\Gamma_l^{(\kappa)}\bigr)+ \Gamma_l^{(\kappa)}/(1+\Gamma_l^{(\kappa)})$ and $\bar{\mathcal{B}}_l^{(\kappa)}\triangleq (\Gamma_l^{(\kappa)})^2/(1+\Gamma_l^{(\kappa)})$.

Finally, we rewrite \eqref{eq:op4wd} as
\begin{subnumcases}{\eqref{eq:op4wd}  \\ \Leftrightarrow}
\sum\nolimits_{l\in\mathcal{L}} \ln(1 + r_l) \leq \frac{C_{m}^{\max}}{\omega_{m}}, \forall m \in \mathcal{M},& \label{eq:op4wdEq1}\\
\frac{\bigl( \Re\{\hat{\textbf{h}}^{H}_{l} (\boldsymbol{\psi}^{(\kappa)}) \boldsymbol{w}_{l}\} \bigr)^2}{r_l} \leq \varphi_l(\boldsymbol{w}|\boldsymbol{\psi}^{(\kappa)}),\ \forall l\in\mathcal{L}                 &\label{eq:op4wdEq2}
\end{subnumcases}
where $\boldsymbol{r}\triangleq\{r_l\}_{l\in\mathcal{L}}$ are newly introduced variables. We note that $\ln(1 + r_l)$ is a concave function and $\varphi_l(\boldsymbol{w}|\boldsymbol{\psi}^{(\kappa)})$ is a convex function. Following the IA principle, constraints \eqref{eq:op4wdEq1} and \eqref{eq:op4wdEq2} are innerly convexified as
\begin{IEEEeqnarray}{rCl}
&&\sum\limits_{l\in\mathcal{L}}\Bigl(\ln(1+r_l^{(\kappa)}) - \frac{r_l^{(\kappa)}}{1+r_l^{(\kappa)}}  + \frac{1}{1+r_l^{(\kappa)}}r_l \Bigr) \leq \frac{C_{m}^{\max}}{\omega_{m}}, \forall m \in \mathcal{M},\label{eq:op4wdConvex1} \qquad
\end{IEEEeqnarray}
\begin{align}
\frac{\bigl( \Re\{\hat{\textbf{h}}^{H}_{l} (\boldsymbol{\psi}^{(\kappa)}) \boldsymbol{w}_{l}\} \bigr)^2}{r_l} \leq \varphi_l^{(\kappa)}(\boldsymbol{w}|\boldsymbol{\psi}^{(\kappa)}),\ \forall l\in\mathcal{L}\label{eq:op4wdConvex2}
\end{align}
where
${\small \varphi_l^{(\kappa)}(\boldsymbol{w}|\boldsymbol{\psi}^{(\kappa)}) \triangleq \sum\nolimits_{j\in\mathcal{L}\setminus  l}\bigl( 2\Re \{ (\boldsymbol{w}_{j}^{(\kappa)})^H\hat{\textbf{h}}_{l} (\boldsymbol{\psi}^{(\kappa)})\hat{\textbf{h}}^{H}_{l} (\boldsymbol{\psi}^{(\kappa)}) \boldsymbol{w}_{j} \} }$ $
 -|\hat{\textbf{h}}^{H}_{l} (\boldsymbol{\psi}^{(\kappa)}) \boldsymbol{w}_{j}^{(\kappa)}|^2\bigr) + \sigma^2.$

The approximate convex problem of \eqref{eq:op4w} solved at  iteration $\kappa+1$ is given as
\begin{subequations} \label{eq:op4wConvex}
	\begin{IEEEeqnarray}{cl}
		\underset{\boldsymbol{w},\boldsymbol{r},\rho}\max\quad  &
		\mathcal{F}^{(\kappa)}(\boldsymbol{w},\boldsymbol{\psi}^{(\kappa)})\triangleq B \sum\nolimits_{l\in\mathcal{L}} \mathcal{F}_l^{(\kappa)}(\boldsymbol{w},\rho|\boldsymbol{\psi}^{(\kappa)}) \label{eq:op4wConvexa}\\ 
		\st\quad & \Re\{\hat{\textbf{h}}^{H}_{l} (\boldsymbol{\psi}^{(\kappa)}) \boldsymbol{w}_{l}\} \geq 0,\ \forall l\in\mathcal{L}, \label{eq:op4wConvexb}\\ 
		& \Omega_l^{(\kappa)}(\boldsymbol{w}|\boldsymbol{\psi}^{(\kappa)}) \geq 0,\ \forall l\in\mathcal{L}, \label{eq:op4wConvexb1}\\
		&  \eqref{eq:op4b},\eqref{eq:op4wa1},  \eqref{eq:op4wcConvex}, \eqref{eq:op4wdConvex1}, \eqref{eq:op4wdConvex2}. \label{eq:op4wConvexc}
	\end{IEEEeqnarray}
\end{subequations}
For given $\boldsymbol{\psi}^{(\kappa)}$, the per-iteration computational complexity of solving \eqref{eq:op4wConvex} is $\mathcal{O}\bigl((4L+M)^{2.5}(L^2(MK+1)^2+4L+M)\bigr)$ \cite{Peaucelle-02-A}.
\subsection{Phase Descent Iteration}
For given $\boldsymbol{w}^{(\kappa+1)}$, by solving \eqref{eq:op4wConvex}, the total power consumption  ${P}_{\Sigma}(\boldsymbol{w}^{(\kappa+1)})$ is fixed and then problem \eqref{eq:op4} with regard to $\boldsymbol{\psi}$ can be expressed as
\begin{subequations} \label{eq:op5psi}
	\begin{IEEEeqnarray}{cl}
		\underset{\boldsymbol{\psi}}\max\quad  &
		\mathcal{G}(\boldsymbol{w}^{(\kappa+1)},\boldsymbol{\psi})\triangleq B \sum_{l\in\mathcal{L}} {R}_{l}(\boldsymbol{\psi}|\boldsymbol{w}^{(\kappa+1)})\label{eq:op5psia}\\ 
		\st\quad 
		&   {R}_{l}(\boldsymbol{\psi}|\boldsymbol{w}^{(\kappa+1)}) \geq {R}^{\rm{min}}_{l}, \forall l \in \mathcal{L}, \label{eq:op5psib}\\ 
		&   \sum\nolimits_{l\in\mathcal{L}} {R}_{l}(\boldsymbol{\psi}|\boldsymbol{w}^{(\kappa+1)}) \leq \frac{C_{m}^{\max}}{\omega_{m}}, \forall m \in \mathcal{M}, \label{eq:op5psic}\\
		& |\psi_{n,r}| = 1, \forall n \in \mathcal{N},  r \in \mathcal{R}.\label{eq:op5psid}
	\end{IEEEeqnarray}
\end{subequations}

The main difficulty  for solving \eqref{eq:op5psi} is due to the unit-modulus constraint \eqref{eq:op5psid}, which is also a nonconvex constraint. To overcome this issue, we relax \eqref{eq:op5psid} by the following convex constraint:
\begin{equation}\label{eq:UMCconvex}
|\psi_{n,r}|^2 \leq 1, \forall n \in \mathcal{N}, r \in \mathcal{R}
\end{equation}
which also implies that
$\sum_{n\in\mathcal{N}}\sum_{r\in\mathcal{R}}|\psi_{n,r}|^2 - NR \leq 0.$
To ensure that constraint \eqref{eq:op5psid} holds true at optimum, we introduce the following theorem.
\begin{theorem}
	The optimality of \eqref{eq:op5psi} is guaranteed by the following penalized optimization problem:
	\begin{subequations}\label{eq:op5psiEqui}
	\begin{align}
		\underset{\boldsymbol{\psi}}\max&\quad  
		 B \sum_{l\in\mathcal{L}} {R}_{l}(\boldsymbol{\psi}|\boldsymbol{w}^{(\kappa+1)}) + \eta\Bigl(\sum_{n\in\mathcal{N}}\sum_{r\in\mathcal{R}}|\psi_{n,r}|^2 - NR\Bigr)\label{eq:theoa}\\
		\st
		& \quad \eqref{eq:op5psib}, \eqref{eq:op5psic}, \eqref{eq:UMCconvex}\label{eq:theob}
	\end{align}
	\end{subequations}
	where $\eta > 0$ is  a constant penalty parameter making the objective and penalty terms comparable.
\end{theorem}
\begin{proof}
Due to constraint \eqref{eq:UMCconvex}, the penalty term $\sum_{n\in\mathcal{N}}\sum_{r\in\mathcal{R}}|\psi_{n,r}|^2 - NR$ is always negative. This allows the uncertainties of the  unit-modulus  constraint to be penalized, which ensures $\psi_{n,r}=1$ at optimum. For a sufficiently large value of $\eta$,
problems \eqref{eq:op5psi} and \eqref{eq:op5psiEqui} share the same optimal solution. A detailed proof can be found in \cite[Appendix C]{HieuTCOM2019}.
\end{proof}

We can see that the developments  presented in Section III-A are very useful to approximate $R_{l}(\boldsymbol{\psi}|\boldsymbol{w}^{(\kappa+1)})$ in the objective \eqref{eq:theoa} and constraints \eqref{eq:op5psib} and  \eqref{eq:op5psic}. We also notice that $\sum_{n\in\mathcal{N}}\sum_{r\in\mathcal{R}}|\psi_{n,r}|^2$ is the sum of quadratic functions, which can be convexified by directly applying the IA method. As a result, we solve the following approximate convex problem of  \eqref{eq:op5psi} at  iteration $\kappa+1$:
\begin{subequations} \label{eq:op5psiConvex}
	\begin{IEEEeqnarray}{cl}
		\underset{\boldsymbol{\psi},\boldsymbol{r}}\max\quad  &
		\mathcal{G}^{(\kappa)}(\boldsymbol{w}^{(\kappa+1)},\boldsymbol{\psi})\triangleq \nonumber\\
		&B \sum\nolimits_{l\in\mathcal{L}} \mathcal{R}_l^{(\kappa)}(\boldsymbol{\psi}|\boldsymbol{w}^{(\kappa+1)}) + \eta\bigl(\mathcal{P}^{(\kappa)}(\boldsymbol{\psi}) - NR\bigr) \label{eq:op5psiConvexa}\quad\\ 
		\st\quad & \mathcal{R}_l^{(\kappa)}(\boldsymbol{\psi}|\boldsymbol{w}^{(\kappa+1)})
		\geq {R}^{\rm{min}}_{l}, \forall l \in \mathcal{L}, \label{eq:op5psiConvexb}\\
		&\frac{|\hat{\textbf{h}}^{H}_{l} (\boldsymbol{\psi}) \boldsymbol{w}_{l}^{(\kappa+1)}|^2}{r_l} \leq \varphi_l^{(\kappa)}(\boldsymbol{\psi}|\boldsymbol{w}^{(\kappa+1)}),\ \forall l\in\mathcal{L},\label{eq:op5psiConvexc}  \\
		& \eqref{eq:op4wdConvex1}, \eqref{eq:UMCconvex}
	\end{IEEEeqnarray}
\end{subequations}
where
$ \mathcal{P}^{(\kappa)}(\boldsymbol{\psi})\triangleq\sum_{n\in\mathcal{N}}\sum_{r\in\mathcal{R}}\bigl(2\Re\{{\bigl(\psi_{n,r}^{(\kappa)}\bigr)}^*\psi_{n,r}\} - |\psi_{n,r}^{(\kappa)}|^2   \bigr) $ and
$\varphi_l^{(\kappa)}(\boldsymbol{\psi}|\boldsymbol{w}^{(\kappa+1)})  \triangleq \sum\nolimits_{j\in\mathcal{L}\setminus  l}\bigl( 2\Re \{ (\boldsymbol{w}_{j}^{(\kappa+1)})^H\hat{\textbf{h}}_{l} (\boldsymbol{\psi}^{(\kappa)})\hat{\textbf{h}}^{H}_{l} (\boldsymbol{\psi})$ $ \boldsymbol{w}_{j}^{(\kappa+1)} \} 
 - |\hat{\textbf{h}}^{H}_{l} (\boldsymbol{\psi}^{(\kappa)}) \boldsymbol{w}_{j}^{(\kappa+1)}|^2\bigr) + \sigma^2. $ The per-iteration computational complexity of solving \eqref{eq:op5psiConvex} is $\mathcal{O}\bigl((2L
 +NR+M)^{2.5}((L+NR)^2+2L
 +NR+M)\bigr)$.

The proposed alternating descent-based iterative algorithm for solving problem  \eqref{eq:op4} is summarized in Algorithm \ref{alg_IterativeAlgorithm}.

\begin{algorithm}
	\begin{algorithmic}[t]
		\protect\caption{Proposed Alternating Descent-based Iterative Algorithm to Solve Problem \eqref{eq:op4}}
		\label{alg_IterativeAlgorithm}
		\global\long\def\algorithmicrequire{\textbf{Initialization:}}
		\REQUIRE Set $\kappa:=0$ and generate an initial feasible point $(\boldsymbol{\psi}^{(0)},\boldsymbol{w}^{(0)})$
		\REPEAT
		\STATE Given $\boldsymbol{\psi}^{(\kappa)}$, solve the convex problem \eqref{eq:op4wConvex} to find the optimal solution $\boldsymbol{w}^\star$ and update $\boldsymbol{w}^{(\kappa+1)} := \boldsymbol{w}^\star$
		\STATE Given $\boldsymbol{w}^{(\kappa+1)}$, solve the convex problem \eqref{eq:op5psiConvex} to find the optimal solution $\boldsymbol{\psi}^\star$ and update $\boldsymbol{\psi}^{(\kappa+1)} := \boldsymbol{\psi}^\star$
		\STATE Set $\kappa:=\kappa+1$
		\UNTIL Convergence\\
		\STATE \textbf{Ouput}: $(\boldsymbol{\psi}^{(\kappa)},\boldsymbol{w}^{(\kappa)})$
	\end{algorithmic} 
\end{algorithm}	

\textit{Convergence analysis:} From \eqref{eq:op4wConvex}, it is clear that $\mathcal{F}(\boldsymbol{w}^{(\kappa+1)},\boldsymbol{\psi}^{(\kappa)}) \geq \mathcal{F}^{(\kappa)}(\boldsymbol{w}^{(\kappa+1)},\boldsymbol{\psi}^{(\kappa)}) \geq \mathcal{F}^{(\kappa)}(\boldsymbol{w}^{(\kappa)},\boldsymbol{\psi}^{(\kappa)}) = \mathcal{F}(\boldsymbol{w}^{(\kappa)},\boldsymbol{\psi}^{(\kappa)})$. Similar to \eqref{eq:op5psiConvex}, we have 
$\mathcal{G}(\boldsymbol{w}^{(\kappa+1)},\boldsymbol{\psi}^{(\kappa+1)}) \geq \mathcal{G}^{(\kappa)}(\boldsymbol{w}^{(\kappa+1)},\boldsymbol{\psi}^{(\kappa+1)})\geq \mathcal{G}^{(\kappa)}(\boldsymbol{w}^{(\kappa+1)},\boldsymbol{\psi}^{(\kappa)}) = \mathcal{G}(\boldsymbol{w}^{(\kappa+1)},\boldsymbol{\psi}^{(\kappa)}).$ As a result, it is true that $\mathcal{E}(\boldsymbol{w}^{(\kappa+1)},\boldsymbol{\psi}^{(\kappa+1)}) \geq \mathcal{E}(\boldsymbol{w}^{(\kappa)},\boldsymbol{\psi}^{(\kappa)})$. In other words, Algorithm \ref{alg_IterativeAlgorithm} generates a sequence $\{(\boldsymbol{w}^{(\kappa)},\boldsymbol{\psi}^{(\kappa)})\}$ of improved points that converges at   least to  a  locally  optimal  solution \cite{Beck:JGO:10}.

\noindent\textit{Choice of $\eta$:} In practice, a very small $\eta$ does not make much difference, leading to a slow convergence. A very large $\eta$ results in an early convergence of Algorithm \ref{alg_IterativeAlgorithm} and a suboptimal solution $\boldsymbol{\psi}^{*}$. Given the simulation setup in Section \ref{sec:Numerical}, we have numerically observed that $\eta=10^3$ ensures the convergence of Algorithm \ref{alg_IterativeAlgorithm} with the highest performance.
\section{Numerical Results}\label{sec:Numerical}


An RIS-CF network including $M=4$ APs, $N=4$ RISs, and $L=8$ UEs is considered as illustrated in Fig. \ref{fig:LayoutConvergence}(a), where all APs, RISs, and UEs are uniformly distributed within a circular region with 1 km radius. The large-scale fading of all channels is modeled as \cite{Ngo:TWC:Mar2017}: $\ds\beta_{a,b} =\ 10^{\frac{{\rm{PL}}(d_{a,b})+\sigma_{sh}z}{10}}$, where $a=\{m,n\}$, $b=\{n,l\}$, $\forall m \in \mathcal{M}$, $n \in \mathcal{N}$, $l \in \mathcal{L}$, and $d_{a,b}$ is the distance (in km) from $a$ to $b$. The shadow fading is modeled as a random variable $z$, which follows $\mathcal{CN} (0,1)$ with standard deviation $\sigma_{sh}$ = 8 dB. The three-slope path loss model (in dB) is considered as \cite{Ngo:TWC:Mar2017}:
$
{\rm{PL}}(d_{a,b}) = -140.7 -35 {\rm{log}}_{10} (d_{a,b}) + 20 a_0 {\rm{log}}_{10} \bigl(d_{a,b}/d_0\bigr) 
 + 15 a_1 {\rm{log}}_{10} \bigl(d_{a,b}/d_1\bigr),
$
where $d_j$, with $j = \left\{0,1\right\}$, represents the reference distance and $a_j = {{\rm{max}}\left\{0,\frac{d_i - d_{a,b}}{|d_i - d_{a,b}|}\right\}}$.  Unless otherwise stated, the key parameters are provided in Table \ref{Parameter}, following studies in \cite{HieuJSAC2020,LuongTSP17,HuangTWC19}. The used convex solver is SeDuMi \cite{Peaucelle-02-A} in the MATLAB environment. 
We compare the performance of Algorithm \ref{alg_IterativeAlgorithm} with three existing resource allocation schemes: $i)$ CF network without RISs, $ii)$  Collocated network with RISs, and $iii)$ Collocated network without RISs. For collocated network, an AP is located at the center of the considered area to serve all UEs. It is equipped with $M K$ antennas and has a maximum transmit power of $M {P}^{\rm{max}}_{m}$. 

\begin{table}[!t]
\footnotesize
\caption{Simulation Parameters.}
\label{Parameter}
\centering
\begin{tabular}{c|c||c|c}
\hline
Parameter & Value & Parameter & Value\\
\hline
\hline
${P}^{\rm{BH}}_{m,l}$ & 0 dBW & $B$ & 20 MHz\\
\hline
${P}_{m}$& 9 dBW & ${P}_{n,r}$ & 10 dBm \\
\hline
${P}_{l}$& 10 dBm & $\xi_{m}$& 1.2\\
\hline
$C_{m}^{\max}\equiv C^{\max},\forall m$ & 500 b/s/Hz & ${P}^{\max}_{m}\equiv {P}^{\max}$& 35 dBm\\
\hline
${R}^{\rm{min}}_{l}$ & 0.5 b/s/Hz & $K$  &  8\\
\hline
$R$ & 8 & $\sigma^2$&  -104 dBm \\
\hline
($d_0$, $d_{1}$)&  (10,50) m & $\eta$ & 10$^3$ \\
\hline
\end{tabular}
\end{table}


\begin{figure}
	\begin{center}
		\begin{subfigure}[System layout used in this section.]{
				\includegraphics[width=0.33\textwidth,trim={0.0cm 0.0cm 0.0cm 0.0cm}]{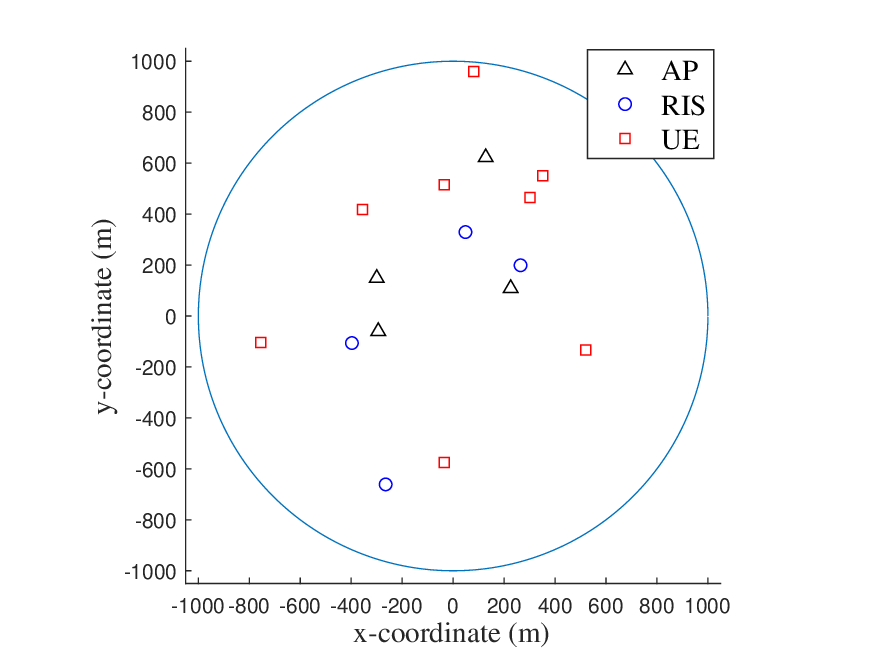}}
			\label{Layout}
		\end{subfigure}
		\begin{subfigure}[Convergence of Algorithm \ref{alg_IterativeAlgorithm} with one random channel realization.]{
				\includegraphics[width=0.35\textwidth,trim={0.0cm 0.0cm 0.0cm 0.00cm}]{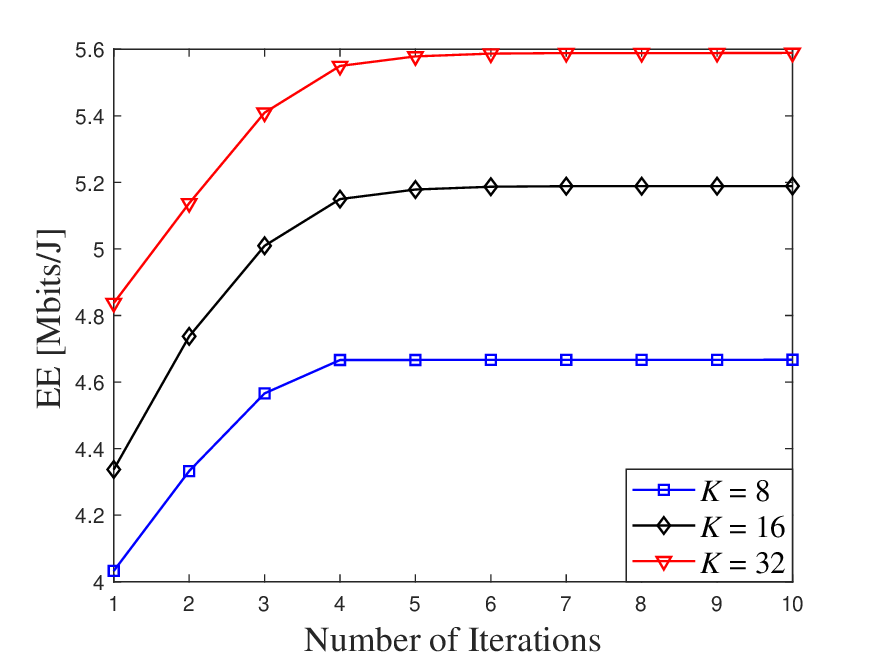}}
			\label{Convergence}
		\end{subfigure}
		\caption{(a) System layout with $M=4$ APs, $N=4$ RISs, and $L=8$ UEs, and (b) Convergence of Algorithm \ref{alg_IterativeAlgorithm} with different number of antennas per AP.}\label{fig:LayoutConvergence}
	\end{center}
\vspace{-20pt}
\end{figure}

Fig. \ref{fig:LayoutConvergence}(b) plots the typical convergence behavior of Algorithm \ref{alg_IterativeAlgorithm} for a random channel realization. On average, Algorithm \ref{alg_IterativeAlgorithm} requires about 6 iterations to reach the almost optimal value of EE in all cases. As expected, increasing $K$ results in better EE, but also requires slightly more iterations.

\begin{figure}
	\begin{center}
		\begin{subfigure}[EE vs. the maximum transmit power per AP, ${P}^{\max}$.]{
				\includegraphics[width=0.35\textwidth,trim={0.0cm 0.0cm 0.0cm 0.0cm}]{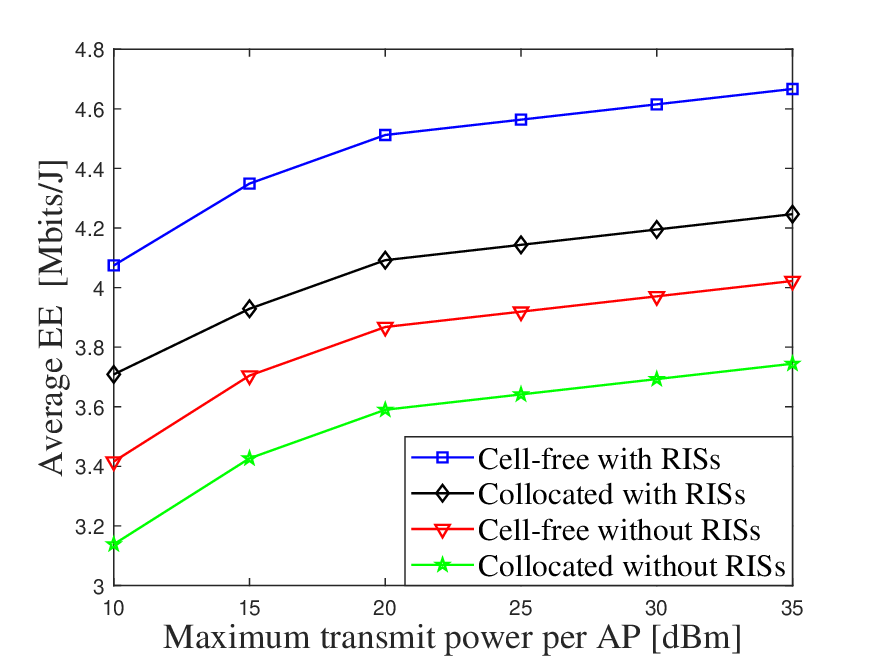}}
			\label{Power}
		\end{subfigure}
		\begin{subfigure}[EE vs. the maximum backhaul capacity, $C^{\max}$.]{
				\includegraphics[width=0.35\textwidth,trim={0.0cm 0.0cm 0.0cm 0.00cm}]{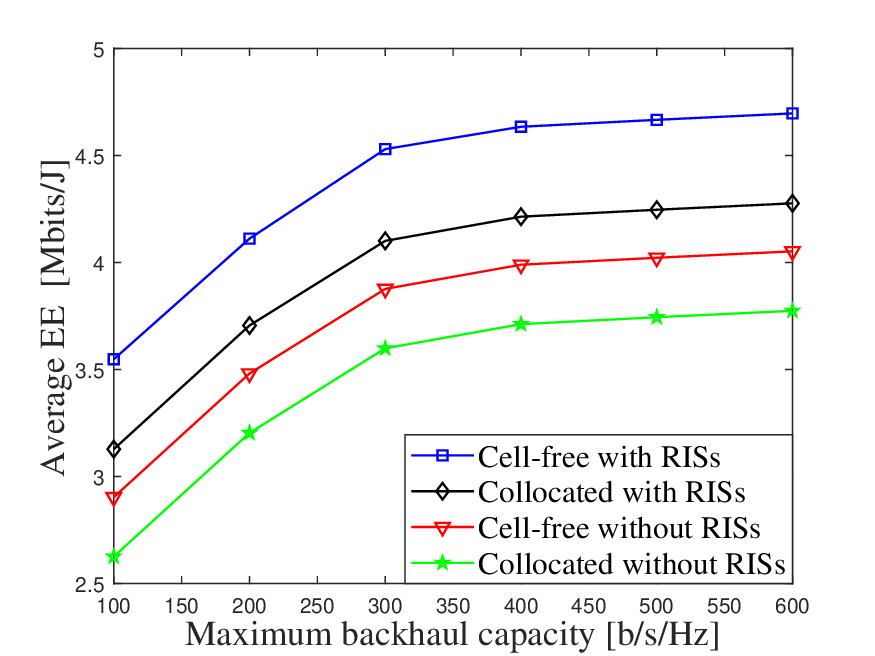}}
			\label{Capacity}
		\end{subfigure}
		\caption{Average EE of the RIS-CF network.}\label{fig:PowerCapacity}
	\end{center}
\vspace{-20pt}
\end{figure}


Fig. \ref{fig:PowerCapacity}(a) depicts the average EE versus the maximum transmit power per AP for different resource allocation schemes. It can be seen that the average EE of all considered schemes significantly enhances when ${P}^{\max}$ increases. Further, the EE of the CF network with and without RISs is much better than that of the collocated network with and without RISs, respectively. This is attributed to the fact that the CF network with distributed APs brings the service antennas closer to UEs, which not only reduces path losses but also provides higher degree of macro-diversity, compared to the collocated network. Moreover, both CF and collocated networks with RISs achieve much higher EE compared to the networks without RISs. This observation confirms that RIS boosts up the EE of CF and collocated networks. Notably, the proposed RIS-CF network provides the best EE among all considered schemes.


In Fig. \ref{fig:PowerCapacity}(b), the average EE is depicted versus the maximum backhaul capacity, $C^{\max}$. As can be seen, the EE of all networks greatly increases when $C^{\max}$ increases. This is because the higher the maximum backhaul capacity, the more data can be conveyed over the backhaul links. Increasing $C^{\max}$ also leads to a remarkable gain in the EE by the proposed RIS-CF over other networks.

\section{Conclusion}
This letter has considered the EE maximization problem of CF networks with the assistance of multiple RISs. The problem involves a joint optimization of  transmit beamformers at  APs and reflecting coefficients at  RISs subject to the limited backhaul capacity constraints, which is formulated as a nonconvex problem. To address this problem, we have developed a low-complexity alternating descent algorithm based on the IA framework, which converges at least to a locally optimal solution. Numerical results have confirmed the fast convergence of the proposed algorithm. Further, they have revealed the advantages of CF and RIS over collocated network.


%

\appendices
\section{Proof of Inequality \eqref{eq:ieObj} }
First, we note that $f(t,z) \triangleq \ln(1+1/t)/z$ is a concave function on the domain $(z > 0,\, t > 0)$ \cite{Dinh:JSAC:18}. By the first-order Taylor approximation, it follows that
\begin{eqnarray}\label{eq:A1}
f(t,z)&\geq& f(t^{(\kappa)},z^{(\kappa)})  - \nabla_{t}f(t^{(\kappa)},z^{(\kappa)})(t-t^{(\kappa)}) \nonumber\\
&&- \nabla_{z}f(t^{(\kappa)},z^{(\kappa)})(z-z^{(\kappa)})\nonumber\\
&=& 2f(t^{(\kappa)},z^{(\kappa)}) + \frac{1}{z^{(\kappa)}(t^{(\kappa)}+1)}  \nonumber\\
&& -\; \frac{1}{z^{(\kappa)}t^{(\kappa)}(t^{(\kappa)}+1)}t - \frac{f(t^{(\kappa)},z^{(\kappa)})}{z^{(\kappa)}}z.
\end{eqnarray}
By replacing $t=y/x^2$ and $t^{(\kappa)}=y^{(\kappa)}/(x^{(\kappa)})^2$, we obtain the inequality \eqref{eq:ieObj}. 

%
%

\ifCLASSOPTIONcaptionsoff
  \newpage
\fi



\bibliographystyle{IEEEtran}
\bibliography{IEEEfull}
\end{document}